\documentclass[letterpaper]{llncs}

\usepackage{amsmath,amssymb,amsfonts}
\usepackage{tgpagella,eulervm}
\usepackage{varioref}
\usepackage{enumerate}

\DeclareMathOperator{\NP}{\mathsf{NP}}

\begin{document}

\title{{\sc On the Grundy number of Cameron graphs}}

\author{
Wing-Kai~Hon\inst{1}
\and 
Ton~Kloks\inst{}
\and 
Fu-Hong~Liu\inst{1}
\and 
Hsiang-Hsuan~Liu\inst{1,2}
\and
Tao-Ming~Wang\inst{3}
}

\institute{National Tsing Hua University, Hsinchu, Taiwan\\
{\tt (wkhon,fhliu,hhliu)@cs.nthu.edu.tw}
\and
University of Liverpool, Liverpool, United Kingdom\\
{\tt hhliu@liverpool.ac.uk}
\and
Tunghai University, Taichung, Taiwan\\
{\tt wang@go.thu.edu.tw}
}

\maketitle 

\begin{abstract}
The Grundy number of a graph is the maximal number of 
colors attained by a first-fit coloring of the graph. 
The class of Cameron graphs is the Seidel switching class of 
cographs. In this paper we show that the Grundy number is 
computable in polynomial time for Cameron graphs. 
\end{abstract}

\section{Introduction}

A proper coloring of a graph is a partition of its vertices 
into independent sets. We refer to the sets in the partition as 
color classes, or simply as colors. 
The chromatic number of a graph $G$, denoted as $\chi(G)$, 
is the minimal number of colors used in a proper coloring. 

\bigskip 

\begin{definition}
Let $\{C_1,\dots,C_k\}$ be the color classes of a proper 
coloring of $G$. The coloring is a first-fit coloring 
if each vertex in color class $C_j$ has at least one 
neighbor in every color class $C_i$ with $i < j$. 
\end{definition}
The maximal number of color classes in a first-fit coloring 
is called the Grundy number of $G$. We denote the Grundy number 
as $\Gamma(G)$. Notice that, if 
$C_1,\dots,C_k$ are the color classes of a first-fit coloring, then 
for each $i$, $C_i$ is a maximal independent set in the 
subgraph induced by 
\[\bigcup_{j=i}^k C_j.\] 
This property characterizes first-fit colorings. 

\bigskip 

A graph is a cograph if it has no induced $P_4$, that is a 
path with four vertices. Cographs are the graphs 
that are closed under unions and joins. It is easily seen 
that in every cograph, every maximal independent set meets 
every maximal clique. This property characterizes the class 
of cographs. By means of this characterization, 
Christen and Selkow prove the following theorem. 
We give a different proof.  

\begin{theorem}
\label{thm 1}
When $G$ is a cograph, 
\[\Gamma(G)=\omega(G)=\chi(G).\]
\end{theorem}
\begin{proof}
Let $G$ be a cograph. If $G$ is the union of two smaller 
cographs, $G_1$ and $G_2$, then 
\begin{equation}
\label{eq1}
\Gamma(G)=\max \;\{\;\Gamma(G_1),\;\Gamma(G_2)\;\}.
\end{equation}

\medskip 

\noindent 
Assume that $G$ is the join of two smaller cographs, 
$G_1$ and $G_2$. Then any independent set has vertices 
only in one of the two graph $G_1$ or $G_2$. It follows 
that 
\begin{equation}
\label{eq2}
\Gamma(G)=\Gamma(G_1)+\Gamma(G_2).
\end{equation}

\medskip 

\noindent 
Notice that the clique number, and also the chromatic number,  
of $G$, satisfy recurrences similar 
to~\eqref{eq1} and~\eqref{eq2}.  
Since the above exhausts all alternatives, this proves the theorem. 
\qed\end{proof}

\bigskip 

\begin{definition}
Let $G$ be a graph and let $S \subseteq V(G)$ be a subset 
of vertices of $G$. The Seidel switch with respect to 
$S$
is the graph obtained from $G$ by complementing the adjacencies 
and nonadjacencies of pairs with one element in $S$ and the 
other in $V \setminus S$. 
\end{definition}
The interest in the Seidel switch grew out of the observation 
that the spectrum, that is, the multiset of eigenvalues  
of the $\{0,-1,1\}$-adjacency matrix remains 
the same under switching. 

\bigskip 

\begin{definition}
A graph is Cameron if it is obtained from a cograph by a 
Seidel switch. 
\end{definition}
The Cameron graphs are perfect. They are 
characterized by a finite set of forbidden 
induced subgraphs, namely the switching class of $C_5$, 
that is, the $C_5$, the bull, gem and co-gem. Another 
characterization states that a graph is Cameron if switching with 
respect to the neighborhood of a vertex produces a cograph 
(with the chosen vertex as an isolated vertex). 
That follows easily from the fact that  
the gem and cogem are forbidden, namely, this 
implies that for each vertex the neighborhood and 
nonneighborhood induce cographs.  
The observation above, together with the linear-time recognition 
of cographs obtained by, eg, Corneil, Perl and Stewart,  
yields an $O(n^2)$ recognition algorithm of Cameron graphs. 

\begin{corollary}
Cameron graphs are recognizable in $O(n^2)$ time. 
\end{corollary}

\bigskip 

Like cographs, Cameron graphs 
form a self-complementary 
class of graphs.  Notice that $C_6$ and $P_5$, but not $P_6$, $C_5$ 
nor any cycle longer than $C_6$, 
are Cameron graphs.  The Grundy numbers of $C_6$, $P_4$ and $P_5$ are 3 so, 
Theorem~\ref{thm 1} is no longer true for Cameron graphs. 

\begin{figure}
\begin{center}
\setlength{\unitlength}{0.7mm}
\thicklines
\begin{picture}(205,20)
\put(10,20){\circle*{2.0}}
\put(15,30){\circle*{2.0}}
\put(20,10){\circle*{2.0}}
\put(25,30){\circle*{2.0}}
\put(30,20){\circle*{2.0}}
\put(10,20){\line(1,2){5}}
\put(10,20){\line(1,-1){10}}
\put(15,30){\line(1,0){10}}
\put(30,20){\line(-1,2){5}}
\put(30,20){\line(-1,-1){10}}
 
 \put(50,30){\circle*{2.0}}
 \put(55,20){\circle*{2.0}}
 \put(65,10){\circle*{2.0}}
 \put(75,20){\circle*{2.0}}
 \put(80,30){\circle*{2.0}}
 \put(55,20){\line(1,-1){10}}
 \put(55,20){\line(1,0){20}}
 \put(55,20){\line(-1,2){5}}
 \put(65,10){\line(1,1){10}}
 \put(75,20){\line(1,2){5}}
 
 \put(100,20){\circle*{2.0}}
 \put(105,30){\circle*{2.0}}
 \put(115,30){\circle*{2.0}}
 \put(110,10){\circle*{2.0}}
 \put(120,20){\circle*{2.0}}
 \put(100,20){\line(1,2){5}}
 \put(105,30){\line(1,0){10}}
 \put(120,20){\line(-1,2){5}}
 \put(110,10){\line(-1,1){10}}
 \put(110,10){\line(-1,4){5}}
 \put(110,10){\line(1,4){5}}
 \put(110,10){\line(1,1){10}}
 
 \put(140,20){\circle*{2.0}}
 \put(145,30){\circle*{2.0}}
 \put(155,30){\circle*{2.0}}
 \put(150,10){\circle*{2.0}}
 \put(160,20){\circle*{2.0}}
 \put(140,20){\line(1,2){5}}
 \put(145,30){\line(1,0){10}}
 \put(160,20){\line(-1,2){5}}
 
 \end{picture}
 \end{center}
 \caption{A $C_5$, bull, gem and cogem}
 \end{figure}
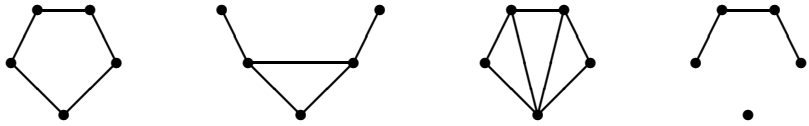
 
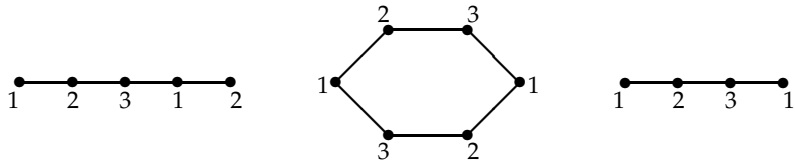
\begin{figure}
\begin{center}
\setlength{\unitlength}{0.7mm}
\thicklines
\begin{picture}(165,40)
\put(10,20){\circle*{2.0}}
\put(20,20){\circle*{2.0}}
\put(30,20){\circle*{2.0}}
\put(40,20){\circle*{2.0}}
\put(50,20){\circle*{2.0}}
\put(10,20){\line(1,0){40}}
\put(7.7,15.2){$1$}
\put(19,15.2){$2$}
\put(29,15.2){$3$}
\put(39,15.2){$1$}
\put(50,15.2){$2$}

\put(70,20){\circle*{2.0}}
\put(80,10){\circle*{2.0}}
\put(80,30){\circle*{2.0}}
\put(95,10){\circle*{2.0}}
\put(95,30){\circle*{2.0}}
\put(105,20){\circle*{2.0}}
\put(70,20){\line(1,-1){10}}
\put(70,20){\line(1,1){10}}
\put(80,10){\line(1,0){15}}
\put(80,30){\line(1,0){15}}
\put(105,20){\line(-1,1){10}}
\put(105,20){\line(-1,-1){10}}
\put(66.5,18){$1$}
\put(78,5.5){$3$}
\put(78,31.5){$2$}
\put(95,5.5){$2$}
\put(95,31.5){$3$}
\put(106.5,18){$1$}

\put(125,20){\circle*{2.0}}
\put(135,20){\circle*{2.0}}
\put(145,20){\circle*{2.0}}
\put(155,20){\circle*{2.0}}
\put(125,20){\line(1,0){30}}
\put(122.7,15.2){$1$}
\put(134,15.2){$2$}
\put(144,15.2){$3$}
\put(155,15.2){$1$}

\end{picture}
\end{center}
\caption{Grundy coloring of $P_4$, $P_5$ and $C_6$}
\end{figure}

\bigskip 

Cographs are characterized by the property that every induced 
subgraph with at least two vertices has a twin, that is, 
a module with two vertices. 
Cameron graphs satisfy a similar characterization.  
Define an anti-twin as a pair of vertices $x$ and $y$ 
such that every other vertex is adjacent to exactly one of the two. 

\begin{theorem}
A graph is Cameron if and only if every induced subgraph 
with at least two vertices has a twin or an anti-twin. 
\end{theorem}

\subsubsection{Some remarks on two-graphs}

A two-graph is a pair $(X,\Delta)$ 
where $X$ is a set and $\Delta$ is a collection of 3-subsets of $X$ 
with the property that every 4-subset of $X$ contains an even number 
of 3-subsets that are in $\Delta$. For example, when $G$ is a graph with 
vertex set $X$, then the collection of triples $\Delta$ that have 
an odd number of edges 
between them (called the `odd triples') 
defines a two-graph $(X,\Delta)$. 
When two graphs are Seidel switching-equivalent 
then they yield the same two-graph and, conversely,  
every two-graph corresponds uniquely to a Seidel switching 
class of graphs.  

\medskip 

\noindent
Consider the triples in $C_5$ with an odd number of edges. 
Cameron calls this two-graph the pentagon. He characterizes the 
two-graphs that do not contain the pentagon as an induced substructure 
as follows. Consider a tree $T$, without vertices of degree two,  
and let $X$ be the set of leaves of $T$. 
Since $T$ is bipartite and connected, it has a unique two-coloring. 
Call the colors in a two-coloring black and white. Let a set of leaves 
$\{x,y,z\}$ be a triple of $\Delta$ if the paths connecting the three 
meet in a black vertex. Then $(X,\Delta)$ is a two-graph and the two-graphs 
obtained in this manner are exactly the two-graphs without the pentagon 
as an induced substructure. 

\medskip 

\noindent 
In the same papers, Cameron characterizes, and counts, also the two-graphs 
that don't have a pentagon nor a hexagon as an induced substructure. 

\newpage 

\section{Cameron graphs}

The following characterization is readily checked. 

\begin{theorem}
A graph is Cameron if and only if there exists a 
coloring of the vertices with colors black and white such that 
%every nontrivial induced subgraph has a  
%partition into two sets subgraphs, $A$ and $B$,
the set of vertices of every nontrivial induced subgraph
has a partition into two sets
such that all crossing 
adjacencies are between vertices of the same color or 
between vertices of opposite colors. 
\end{theorem}
When a Cameron graph is obtained from a cograph 
via the switching with respect to a set $S$, then 
coloring all vertices of $S$ white and the remaining 
vertices black, satisfies the property mentioned in the 
theorem. For example, let $G$ is a Cameron graph and let $x$ 
be any vertex of $G$. Color 
the vertices of $N(x)$, 
that is, the neighborhood of $x$, white, and the remaining 
vertices black. If we switch the graph with respect to the 
white vertices, we obtain a cograph $H$ in which $x$ is an 
isolated vertex. Since $H$ is a cograph it has a binary, rooted  
decomposition tree, called a cotree. Each leaf of the tree 
corresponds uniquely to a vertex of the graph. Each internal node 
(including the root) of this 
tree is labeled as a join node or a union node. When the node is a 
join node, each vertex mapped to a leaf in the left subtree is made 
adjacent to each vertex that is mapped to a leaf in the right subtree. 
When an internal node of the decomposition tree is labeled as a union 
node, then no vertex of the left subtree is adjacent to any vertex of 
the right subtree. 

\medskip 

The decomposition tree for the Cameron graph $G$ is the same rooted binary tree. 
When an internal node is labeled as a join node for the cograph, then 
vertices of similar colors in the left and right tree are made adjacent. 
When an internal node is labeled as a union node, then exactly 
those pairs of vertices in the left and 
right subtree that have opposite colors are made adjacent. 
Henceforth, we refer to the join nodes in the decomposition tree as 
`parallel nodes,' and to the union nodes as `crossing nodes.' 

\section{Example}
\label{section example}

As an example, we present the following real-world problem. In some 
faraway country there are $N$ villages, numbered $1,\dots,N$. 
In village $i$ there are $b_i$ boys and $g_i$ girls eligible for marriage. 
However, albeit a bit archaic, 
the country's law and culture forbids the marriage of girls and boys 
that are from the same village. To study population growth, 
scientists are interested in this question: what is the 
{\em minimal\/} number of couples that get married, if we don't allow 
any single boy-and-girl pair from distinct villages. 

\bigskip 

The Cameron graph $G$ that represents the problem consists of 
$N$ cliques. Clique $i$ consists of $b_i$ black vertices and 
$g_i$ white vertices. Between any two cliques $i$ and $j$ 
we have a parallel 
connection, that is, all black vertices of clique $i$ are 
adjacent to all black vertices of clique $j$ and all white 
vertices of clique $i$ are adjacent to all white vertices of 
clique $j$. (Only heterosexual marriages are allowed. Thus, the 
Cameron graph is cobipartite.)  

\bigskip 

We are interested in the Grundy number of this Cameron graph. 
Notice that each independent set consists either of a single 
vertex or, of a black and white vertex from different cliques. 
The minimal number of pairs that get married is therefore, 
\begin{equation}
|V(G)|-\Gamma(G).
\end{equation}
It remains to show that we can compute $\Gamma(G)$ in polynomial time. 

\bigskip 

Construct the decomposition tree for $G$; this is a rooted binary tree 
with the black and white vertices in the leaves. Each internal 
vertex is labeled as a parallel node or a crossing node. 
In this example, we can have a decomposition tree 
with one crossing node for each village,  
that node connects all the boys and girls from that same 
village. We may assume that all the other internal nodes are 
parallel nodes. 

\bigskip 

Our method to solve this problem is a dynamic programming on this decomposition tree. 
For an internal node $t$ denote the set of vertices mapped to the 
leaves in the subtree by $V_t$. The algorithm computes a boolean function 
\[\tau(b_t,g_t,b_t^{\prime},g_t^{\prime}),\]
which is ${\tt true}$ if 
\begin{enumerate}[\rm (i)]
\item exactly $b_t$ single boys and $g_t$ single girls from $V_t$  
will, presumably, get married with boys and girls from $V(G)\setminus V_t$ 
(they marry in the future), and 
\item $b_t^{\prime}$ boys and $g_t^{\prime}$ girls from $V_t$ 
stay single altogether. 
\end{enumerate}

\bigskip 

The computation of the function $\tau$ for each village is easy; in the remainder 
we consider internal nodes of the decomposition tree such that each village 
is either fully contained in the leaves of the left subtree, or it is 
fully contained in the leaves of the right subtree, or it has an 
empty intersection with the leaves in the subtree.\footnote{That is, the decomposition 
tree represents a laminar family of subsets with the villages as atoms.} 

\bigskip 

We may assume the following \underline{principle of optimality}; for any node 
$t$ in the decomposition tree, 
either $b_t^{\prime}=0$ or $g_t^{\prime}=0$ or 
all these boys and girls are from a single village. For parameters 
such that this condition 
cannot be fulfilled, we let 
\[\tau=\text{\tt false}.\]   

\bigskip 

Consider an internal node $t$, and consider a parameter set 
$\{b_{\ell},g_{\ell},b_{\ell}^{\prime},g_{\ell}^{\prime}\}$ for the left subtree and a 
parameter set  $\{b_r,g_r,b_r^{\prime},g_r^{\prime}\}$ for the right subtree. We assume  
that the function $\tau$ evaluates as $\text{\tt true}$ for the parameters in the left 
subtree and for the parameters in the right subtree. 
The resulting set of parameters 
\[\{\;b,\;g,\;b^{\prime},\;g^{\prime}\;\}\] 
for the node $t$ is then
valid if there exist numbers $\alpha$ and $\beta$ such that 
$\alpha$ boys on the left get married to $\alpha$ girls on the right and 
$\beta$ girls on the left get married to $\beta$ boys on the right. Thus, 
these numbers must satisfy  
\[0 \leq \alpha \leq \min\;\{\;b_{\ell},\;g_r\;\} \quad\text{and}\quad  
0\leq \beta \leq \min\;\{\;g_{\ell},\;b_r\;\},\] 
and the 
resulting set of parameters then equals 
\begin{enumerate}
\item $b=b_r+b_{\ell}-\alpha-\beta$, 
\item $g=g_r+g_{\ell}-\alpha-\beta$,  
\item $b^{\prime}=b_{\ell}^{\prime}+b_r^{\prime}$ and 
\item $g^{\prime}=g_{\ell}^{\prime}+g_r^{\prime}$.  
\end{enumerate} 

\bigskip 

The optimality condition requires that 
$b^{\prime}=0$ or $g^{\prime}=0$ or that these singles are 
all from the same village. By the assumption that each village is 
either fully contained in a subtree, or disjoint from that subtree, 
this implies that we must have  
\begin{enumerate}[\rm (a)]
\item $b^{\prime}=b_{\ell}^{\prime}$ and $g^{\prime}=g_{\ell}^{\prime}$ and 
$b_r^{\prime}=g_r^{\prime}=0$, or
\item $b^{\prime}=b_r^{\prime}$ and $g^{\prime}=g_r^{\prime}$ and 
$g_{\ell}^{\prime}=b_{\ell}^{\prime}=0$, or 
\item $b^{\prime}=b_{\ell}^{\prime}+b_r^{\prime}$ and 
$g^{\prime}=g_{\ell}^{\prime}=g_r^{\prime}=0$, or 
\item $g^{\prime}=g_{\ell}^{\prime}+g_r^{\prime}$ and 
$b^{\prime}=b_{\ell}^{\prime}=b_r^{\prime}=0$.
\end{enumerate}

\bigskip 

At the root of the decomposition tree we require that 
\[b=g=0,\] 
because these singles won't have any opportunity to marry in the future, 
that is, with singles outside $V(G)$. 
The answer to the problem, that is, the minimal number of married couples is 
therefore 
\[\min \;\left\{\; \frac{|V(G)|-b^{\prime}-g^{\prime}}{2}\;
\mid \; \tau(0,0,b^{\prime},g^{\prime})=
\text{\tt true}\;\right\}.\] 

This shows that the Grundy number for Cameron graphs in 
this example is computable in polynomial time. In the following 
section we discuss the general case. 

\section{The Grundy number of Cameron graphs}

Let $G$ be a Cameron graph with a black-and-white coloring and let 
$H$ be the cograph that results from the Seidel switch of $G$ with respect 
to the set of white vertices. Notice that we may assume that $G$ is connected 
and that $\Bar{G}$ is connected, since otherwise $G$ is a cographs and 
we are done, by Theorem~\vref{thm 1}. 

\bigskip 

Consider an internal node $t$ of the decomposition tree. We refer to 
$V_t$ as the vertices that are mapped to leaves in the subtree rooted 
at $t$. The set of black and white vertices of $V_t$ are denoted as 
$B_t$ and $W_t$. 
similar as in the example, we store information of partial colorings 
of an internal node $t$ in a 
boolean function 
\[\tau(m_t,b_t,w_t,b_t^{\prime},w_t^{\prime}).\]  
Here 
\begin{enumerate}
\item $m_t$ is the number of mixed color classes, that is, 
color classes of $G[V_t]$ that contain at least one black 
and one white vertex; 
\item $b_t$ is the number of color classes that consist of only 
black vertices, and that will, presumably, `marry' (unite with 
white color classes) in the future; 
\item $w_t$ is the number of color classes that consist of only white 
vertices and that will, presumably, marry in the future; 
\item $b_t^{\prime}$ is the number of black color classes that will 
stay forever single, and 
\item $w_t^{\prime}$ is the number of white color classes that will 
stay forever single. 
\end{enumerate}

\bigskip 

We assume that there is an ordering of these color classes such that 
the subsets of the colors on $G[B_t]$ and $G[W_t]$ form a first-fit coloring. 
Consider the case where the first color class $C$ is of mixed type. 
Then $C \cap B_t$ and $C \cap W_t$ are maximal independent sets in 
$G[B_t]$ and $G[W_t]$. Notice that this implies that $C$ is a maximal 
independent set in $G$ (because every internal node of the decomposition 
tree is a parallel node or a crossing node). 

\bigskip 

We discuss the updating procedures. Consider a \underline{crossing node} $t$. 
Let 
\[\{\;m_{\ell},\;b_{\ell},\;w_{\ell},\;b_{\ell}^{\prime},\;w_{\ell}^{\prime}\;\} \quad\text{and}\quad 
\{\;m_r,\;b_r,\;w_r,\;b_r^{\prime},\;w_r^{\prime}\;\}\] 
be 
sets of parameters for the left and right child for which $\tau$ 
evaluates as {\tt true}. 
Then, if we denote the set parameters of $t$ by $m$, $b$, $w$, $b^{\prime}$ 
and $w^{\prime}$ we have 
\begin{enumerate}
\item $m=m_{\ell}+m_r$, 
\item $b=\max\;\{\;b_{\ell},\;b_r\;\}$, 
\item $w=\max \;\{\;w_{\ell},\;w_r\;\}$, 
\item $b^{\prime}=\max\;\{\;b_{\ell}^{\prime},\;b_r^{\prime}\;\}$, 
\item $w^{\prime}=\max\;\{\;w_{\ell}^{\prime},\;w_r^{\prime}\;\}$. 
\end{enumerate}

\bigskip 

For the \underline{parallel nodes}, the updates are similar as in the 
example of Section~\ref{section example}. 
Assume that $\alpha$ black color classes on the left marry to 
$\alpha$ white color classes on the right and that $\beta$ white color 
classes on the left marry with $\beta$ black color classes on the right. 
Then a necessary condition for the numbers $\alpha$ and $\beta$ is that 
\[\alpha \leq \min\;\{\;b_{\ell},\;w_r\;\} \quad\text{and}\quad 
\beta \leq \min\;\{\;w_{\ell},\;b_r\;\}.\] 
The new parameters become 
\begin{enumerate}
\item $m=m_{\ell}+m_r+\alpha+\beta$, 
\item $b=b_{\ell}+b_r-\alpha-\beta$, 
\item $w=w_{\ell}+w_r-\alpha-\beta$, 
\item $b^{\prime}=b_{\ell}^{\prime}+b_r^{\prime}$, and
\item $w^{\prime}=w_{\ell}^{\prime}+w_r^{\prime}$. 
\end{enumerate}

\bigskip 

The optimality condition requires that $b^{\prime}=0$ or that $w^{\prime}=0$ 
or that all these monochromatic color classes are from the same 
parallel component. That is, we must have (similar as in the example): 
\begin{enumerate}[\rm (a)]
\item $b^{\prime}=b_{\ell}^{\prime}$ and $g^{\prime}=g_{\ell}^{\prime}$ and 
$b_r^{\prime}=g_r^{\prime}=0$, or
\item $b^{\prime}=b_r^{\prime}$ and $g^{\prime}=g_r^{\prime}$ and 
$g_{\ell}^{\prime}=b_{\ell}^{\prime}=0$, or 
\item $b^{\prime}=b_{\ell}^{\prime}+b_r^{\prime}$ and 
$g^{\prime}=g_{\ell}^{\prime}=g_r^{\prime}=0$, or 
\item $g^{\prime}=g_{\ell}^{\prime}+g_r^{\prime}$ and 
$b^{\prime}=b_{\ell}^{\prime}=b_r^{\prime}=0$.
\end{enumerate}

\bigskip 

At the \underline{root}, we are only interested in the sets of parameters 
with 
\[b=w=0,\] 
since there is no opportunity for these 
color classes to marry in the future, and so, the Grundy 
number is, 
\[\Gamma(G)=\max\;\{\; m+b^{\prime}+w^{\prime} \;|\; 
\tau(m,0,0,b^{\prime},w^{\prime})=\text{\tt true}\;\}.\]

\bigskip 

\begin{theorem}
There exists a polynomial-time algorithm that computes the 
Grundy number of Cameron graphs. 
\end{theorem}
\begin{proof}
We prove first that there is a first-fit coloring with the computed 
set of parameters. 

\medskip 

\noindent 
Consider a crossing node $t$. Let parameters for the left and right 
subtree be 
\begin{equation}
\label{eqn1}
\{\;m_{\ell},\;b_{\ell},\;w_{\ell},\;b_{\ell}^{\prime},\;w_{\ell}^{\prime}\;\} 
\quad\text{and}\quad 
\{\;m_r,\;b_r,\;w_r,\;b_r^{\prime},\;w_r^{\prime}\;\}.
\end{equation}
We may assume that there are partial first-fit colorings for the 
graphs in the left and right subtree. The parameter setting for the node $t$ 
is 
\begin{enumerate}[\rm (1)]
\item $m=m_{\ell}+m_r$; 
\item $b=\max\;\{\;b_{\ell},\;b_r\;\}$; 
\item $w=\max\;\{\;w_{\ell},\;w_r\;\}$;
\item $b^{\prime}=\max\;\{\;b_{\ell}^{\prime},\;b_r^{\prime}\;\}$;
\item $w^{\prime}=\max\;\{\;w_{\ell}^{\prime},\;w_r^{\prime}\;\}$. 
\end{enumerate}
We claim that that there are colorings for the subgraphs on the left 
and right that start with the mixed color classes. We prove that below. 
Then the claim follows easily; any mixed color class that starts a 
coloring on the left or right is a maximal independent set in $G[V_t]$. 
This proves that there exists a coloring for $G[V_t]$ that starts with 
$m$ mixed color classes. 

\medskip 

\noindent 
Consider removing all mixed color classes. The remaining graphs on the left and 
right, induced by the black and white vertices are cographs. By Theorem~\ref{thm 1}, 
the number of monochromatic color classes is equal to the chromatic number of the 
respective cographs. For the node $t$, after removal of the mixed color classes, 
the formulas are the formulas that compute the chromatic numbers of the black 
and white cographs (which are united, since $t$ is a crossing node; so the formulas 
are given by~\eqref{eq1} on Page~\pageref{eq1}). 

\medskip 

\noindent 
Consider a parallel node $t$. Consider colorings for the left and right subgraph as 
in Equation~\eqref{eqn1}. 
By induction, there exist colorings for the left and right subgraph that 
start with the mixed color classes (if any). 
Consider removing the mixed color classes and let 
$V_t^{\prime}$ be the remaining set of vertices of $V_t$. 
Consider numbers $\alpha$ and 
$\beta$ with 
\[\alpha \leq \min\;\{\;b_{\ell},\;w_r\;\} \quad\text{and}\quad 
\beta \leq \min\;\{\;w_{\ell},\;b_r\;\}.\] 
Take the first $\alpha$ black color classes of the graph on the left 
and unite them with the first $\alpha$ white color classes on the right. 
This produces $\alpha$ new, mixed color classes. Notice that these can start 
a first-fit coloring in graph induced by $V_t^{\prime}$. 

\medskip 

\noindent 
It is readily checked that we may assume the optimality condition. 
There can be no unmarried black and white pair, which 
are in different parallel components. 

\medskip 

\noindent 
Consider a first-fit coloring of $G$. 
Consider a node $t$; by induction we may assume that the coloring 
induces partial colorings for the left and right subgraph. That is, 
the mixed color classes start the coloring on the left and right, and 
the monochromatic color classes form a coloring of the cographs,  
induced by the black and white vertices, with $\chi$ colors. If the node 
$t$ is crossing, the monochromatic color classes need to unite, since 
otherwise they are not maximal. In case $t$ is a parallel node, notice that 
the monochromatic color classes on the left and right need to be 
maximal, otherwise they do not form proper first-fit color classes. 

\medskip 

\noindent 
This proves the theorem. 
\qed\end{proof}

\newpage 

\section{Supplementary information on Grundy colorings}

A complete coloring of a graph is a coloring such that 
for every pair of colors, there exists an edge whose endpoints' colors  
match the colors of the pair. In other words, the union of no 
two color classes is an independent set. 
The maximal number of colors in a complete coloring is 
called the achromatic number and it is usually 
denoted as $\Psi(G)$.  This coloring owes its name, `complete' coloring, to   
the homomorphism $G \rightarrow K_k$, where $k=\Psi(G)$ is the maximal $k$ 
for which such a homomorphism exists. 
Computing the achromatic number is 
$\NP$-complete, even for trees and for trivially perfect graphs 
(which are the graphs without induced $P_4$ and $C_4$; so 
they include the cographs). 
Obviously, we have 
\[\chi(G) \leq \Gamma(G) \leq \Psi(G).\] 
Interestingly, the achromatic number is fixed-parameter 
tractable, that is, there exists a constant $c$ and a function 
$f: \mathbb{N} \rightarrow \mathbb{N}$ such that for each 
$k \in \mathbb{N}$, the question whether $\Psi(G) \geq k$ can 
be decided in $O(f(k)\cdot n^c)$ time. As far as we know, 
the question whether the Grundy number is fixed-parameter tractable 
is open. 

\bigskip 

Zaker showed that 
any graph with $\Gamma(G) \geq k$ has an induced 
subgraph with at most $2^{k-1}$ vertices 
and with Grundy number at least $k$.  This is called a $k$-witness. 
The existence of a $k$-witness implies that there is an algorithm 
that runs in $O(n^{2^{k-1}})$ time to decide if $\Gamma(G) \geq k$. 

\bigskip 

Zaker shows that computing $\Gamma(G)$ is $\NP$-complete for 
co-bipartite graphs. 
Obviously, 
$\Gamma(G) \leq \Delta(G)+1$. According to Havet and Sampaio, 
deciding whether $\Gamma(G) \leq \Delta(G)$ is $\NP$-complete for bipartite graphs. 
Deciding if $\Gamma(G) \geq \Delta(G)-k$ is fixed-parameter tractable 
with respect to the parameter $k$. 
Bonnet et al. show that the Grundy 
number is fixed-parameter tractable for chordal graphs, claw-free graphs and 
graphs excluding a fixed minor. 

\bigskip 

For cographs, $\Gamma(G)=\chi(G)$. Zaker showed that deciding whether 
$\Gamma(G)=\chi(G)$ is co-$\NP$-complete. 
Tang et al. mention the following conjecture by Zaker, which we repeat 
here because we think it is interesting (apparently, Zaker did not publish 
this conjecture). 

\begin{conjecture}
If $G$ is $C_4$-free then $\Gamma(G) \geq \delta(G)+1$, 
where $\delta(G)$ is the minimal degree of $G$. 
\end{conjecture}

\end{document}